\documentclass[a4paper]{article}
\usepackage{fullpage,url,amsmath,amsfonts,amssymb, tedmath}%authblk}
\newcommand{\tr}{\text{tr}\,}
\newcommand{\cc}{\mathbb{C}}
\newcommand{\cir}{\,\text{circ}}

\title{Systems of MDS codes from units and idempotents.}
\author{Barry Hurley\footnote{National University
 of Ireland Galway, email: Barryj\_2001@yahoo.co.uk} \, \& Ted
 Hurley\footnote{National Universiy of Ireland Galway, email:
 Ted.Hurley@Nuigalway.ie }}
\date{} 
\begin{document}

\maketitle

\begin{centering}

\end{centering}

\begin{abstract} Algebraic systems are constructed from which series of 
maximum distance separable (mds) codes are derived. 
The methods use unit  and idempotent schemes. 
%and their related
%Fourier type matrices. 
\end{abstract}

\section{Introduction}

Algebraic coding theory deals with the design of error-correcting and
error-detecting codes for the reliable transmission of information
across noisy channels. It has many applications, modern communications
could not be undertaken without it and much research is still
on-going.  Coding theory in general makes use of many abstract notions
such as fields, group theory, polynomial algebra and areas of discrete
mathematics. % especially number theory and the theory of experimental
%designs.

A basic reference for coding theory is Blahut
\cite{blahut}. Codes
from zero divisors and unit-derived codes in group rings and matrix
rings are obtained in \cite{hur1} and in more detail in \cite{hur2}.
  An $(n,r,d)$ (linear) code is a code of length $n$, dimension $r$
  and distance $d$. By the Singleton Bound, see for example Theorem
  3.2.6 of \cite{blahut}, the maximum $d$ can have is $(n-r+1)$ and so
  an {\em mds (maximum distance separable) code} is defined as a code
  of the form $(n,r,n-r+1)$ or equivalently a code of the form
  $(n,n-r,r+1)$.

Here systems and series of such mds codes are derived. %From a single
%$n\ti n$
 This paper originated from ideas of constructing codes from
{\em complete orthogonal sets of idempotents} in general rings and in
particular in group rings.
Constructions of such idempotent systems in general are dealt with in
\cite{hur5} where these are used to construct paraunitary (single and
multivariable) matrices which are used in the communications' areas.

A query to the pub-group forum was answered by Marty Isaacs who brought the
results of \cite{isaacs1} and a result of Chebotar\"ev to our
attention. Using Chebotar\"ev's result directly and the unit-derived
coding method of \cite{hur1} enables the construction of series of mds
codes over $\cc$ initially using the Fourier matrices. Results in %The
\cite{isaacs1}  
are then exploited to construct finite fields over which the
Chebotar\"ev's result is true and hence to derive series of mds codes
over these finite fields.
The paper \cite{isaacs1} in addition contains a proof of
Chebotar\"ev's original result and a number of other nice results
besides.

In section \ref{above} methods are
derived for constructing general codes from complete orthogonal sets
of idempotents. 
Specialising then enables systems of mds codes to be derived over
various fields; Chebotar\"ev's result and the results of
\cite{isaacs1} are used to show algebraically that the maximum
distances are actually attained.

Sets of vectors $S=\{e_0,e_1,\ldots, e_{n-1}\}$ in 
 $K^n$ for various fields $K$ and prime $n$ are derived  
such that any $r$ elements of $S$ generates an  $(n,r,n-r+1)$
code. For given $r$ there are ${n\choose r}$ choices for defining such
a code from $S$ and each code is different.
 
Sets of idempotents matrices $T=\{E_0,E_1,\ldots,E_{s-1}\}$ in $K_{n\ti n}$ are
defined over fields $K$ such that $\{E_j \, | j\in J\}$ where
$J\subset I= \{0,1,2,\ldots, s-1\}$ generates an $(n,r)$ code where $r
= \sum_{j\in J}\rank E_j$.  In certain cases when $s=n$ and $n$ is prime 
these are shown to be mds
codes. %When the $E_i$ are circulant for example 
 %the codes are also circulant (cyclic). Cases when $s=n$ and the $E_i$ are of
 %$\rank 1$ and circulant are produced so that the code generated by
 %any $\{E_j \, | \, j\in J\}$ is an mds $(n,r,n-r+1)$ code where $|J|=r$.    

The mds codes derived using idempotents from the cyclic group ring 
 may be considered as
those where the Fourier transform has zeros at $k$ specified locations
which need not be consecutive.  %the codes derived here mds codes.

%Generator and check
%matrices in all cases are immediate from the constructions.

One of the features of some of the
series of mds codes derived is that these are codes over  
a finite field $F_p$, for $p$  a prime,  and modular arithmetic may
be used.       

Section \ref{decode} considers decoding methods for such codes. 
As the dimension and distance of a space generated by a
 subset of $S$ is easily determined, it is then %this enables good decoding methods 
 possible to find  $t$-error correcting pairs in many of 
 these $(n,r,n-r+1)$ codes 
 for maximum $t$ (that is
 for $t=\lfloor\frac{(n-r)}{2}\rfloor$).
Now $t$-error correcting pairs were introduced by Duursma and
K\"otter,\cite{koetter} and by Pellikaan \cite{pell}.

\section{Codes from units}\label{units} 
Unit-derived codes, as in \cite{hur1,hur2}, are defined as follows.
Suppose $UV=I$ in $F_{n\ti n}$.  Divide $U = \begin{pmatrix} A
  \\ B \end{pmatrix}$ into block matrices where $A$ is an $r\times n$
matrix and
$B$ is $(n-r)\times n$. Similarly divide $V$ into blocks $V=
\begin{pmatrix} C & D\end{pmatrix}$ where $C$ is an $n \times r$
  matrix and $D$ is an $n\times (n-r)$ matrix.

Now $AD = 0$ as $UV=I$. It is easy to show that $A$ generates an
$(n,r)$ code and that $D\T$ is a check matrix for this code.

The above method is generalised as follows, see \cite{hur1,hur2} for
details.  Let the rows of $U$ be denoted by $\{u_1,u_2,\ldots, u_n\}$
and the columns of $V$ denoted by $\{v_1,v_2,\ldots, v_n\}$. Choose
$r$ rows $\{u_{i_1},u_{i_2},\ldots,u_{i_r}\}$ of $U$ as a generating
matrix $A$ which is then of size $r\ti n$ and has $\rank r$.  Let
$K=\{1,2,\ldots,n\}$ and $L=\{i_1,i_2,\ldots, i_r\}$ and $J=
(K-M)$. Choose $D$ to be the matrix formed (in any order) from the
(column) vectors $S=\{v_j\, | \, j\in J\}$. Then $D$ has $\rank (n-r)$
and is of size $n\ti (n-r)$ and $D\T$ is a check matrix for the $(n,r)$
code generated by $A$.

(The $r$ rows of $U$ used to form $A$ are usually taken in their
naturally occurring order but this is not necessary.  The matrix $D$
can be formed from the column vectors $S$ in any order but the natural
order of the elements of $S$ would normally be used.)

These codes are linear but in general are not ideals.

Thus any rows of $U$ may be used as a generator matrix for a code and
then corresponding columns of $V$ as indicated give a check
matrix. From a single unit of size $n\ti n$ there are ${n\choose r}$
choices for an $(n,r)$ code and each code is different.
The fact that the codes are different follows from the following Lemma
\ref{differ}. Define, in a vector space, $\langle X\rangle $ to be the
subspace generated by $X$. 
\begin{lemma}\label{differ} Let $T$ be a set of linearly independent vectors and
  $S\subseteq T, W\subseteq T$. Then $\langle S \rangle \cap \langle W
  \rangle = \langle S\cap W \rangle$.
\end{lemma}
\begin{proof} The proof follows directly from the linearly
  independence of the sets $S$ and $W$.
\end{proof} 

Suppose then $UV=1$ in $F_{n\ti n}$. Then taking any $r$ rows of $U$
as a generator matrix $U_r$ and then certain defined $(n-r)$ columns
of $V$ to give the check matrix $V_{n-r}$ defines an $(n,r)$ code. Let
such a code be denoted by $\mathcal{C}_r$.
If matrix $V$ has the property that the determinant of any square
submatrix of $V$ is non-zero then any such code is an mds
$(n,r,n-r+1)$ code.

\begin{theorem}\label{bigdist}
 Suppose the determinant of any square submatrix of $V$ is
 non-zero. Then any such code $\mathcal{C}_r$ has distance $(n-r+1)$
 and is thus an $(n,r,n-r+1)$ mds code.
\end{theorem}
\begin{proof} The proof follows from Theorem 3.2.2/Corollary 3.2.3 of
  \cite{blahut} as any $(n-r)\ti (n-r)$ submatrix of $V$ has non-zero
  determinant. 
\end{proof}

%In Section \ref{7} series of cyclic mds codes obtained are obtained
%using complete orthogonal sets of idempotents.

Suppose for example $UV=I$, $U$ has size $101\ti 101$ and we are
interested in $(101,50)$ codes. Choosing {\em any} $50$ of the rows of
$U$ gives such a code and each one is different thus giving $101
\choose 50$ such codes. Now $101 \choose 50$ is of order $10^{29} $ or
$2^{97}$. There exist $101\choose 80$, which is of order of $2^{37}$,
high rate code $(101,80)$ in such a system.
If the determinant of any square submatrix of $V$ is non-zero we get
of the order of $2^{97}$ mds codes $(101,50,52)$ and of the order of
$2^{37}$  mds codes $(101,80,22)$.

\section{Chebotar\"ev's Theorem}\label{8}
Let $\al$ be a primitive $n^{th} $ root of unity in a field $K$ in
which the inverse of $n$ exists. The Fourier $n\ti n$ matrix $F_n$
over $K$ is

$F_n=\begin{pmatrix} 1& 1& 1 & \ldots & 1 \\ 1 & \al & \al^2 & \ldots
& \al^{(n-1)} \\ 1 & \al^2 & \al^4 & \ldots & \al^{2(n-1)} \\ \vdots &
\vdots & \vdots & \vdots & \vdots \\ 1 & \al^{n-1}& \al^{2(n-1)} &
\ldots & \al^{(n-1)(n-1)} \end{pmatrix}$.

The inverse of $F_n$ is

$F_n^*=\frac{1}{n}\begin{pmatrix} 1& 1& 1 & \ldots & 1 \\ 1 & \al^{-1} &
  \al^{-2} & \ldots & \al^{-(n-1)} \\ 1 & \al^{-2} & \al^{-4} & \ldots
  & \al^{-2(n-1)} \\ \vdots & \vdots & \vdots & \vdots & \vdots \\ 1 &
  \al^{-(n-1)}& \al^{-2(n-1)} & \ldots &
  \al^{-(n-1)(n-1)} \end{pmatrix} $.

($\al^{-1}$ is a primitive $n^{th}$ root of $1$ also and 
the matrix $nF_n^*$ is considered  a Fourier matrix over $K$.) 

\subsection{Fourier}\label{fourier}
We are grateful to Marty Isaacs for bringing the following result of
Chebotar\"ev and the paper \cite{isaacs1} to our attention. A proof of
this Chebotar\"ev theorem may be found in \cite{isaacs} and proofs
also appear in the expository paper of P.Stevenhagen and H.W Lenstra
\cite{steven}; paper \cite{simple} contains a relatively short
proof. There are several other proofs in the literature some of which
are referred to in \cite{steven}.  A proof of the Theorem is also
contained in
\cite{isaacs1} and this paper contains many nice related results and
results related to fields in general (and not just $\cc, \R$) 
as we shall see later. Paper
\cite{tao} contains a proof of Chebotar\"ev's theorem and refers to it
as `an uncertainty principle'. % for finite cyclic groups'.
\begin{theorem}\label{cheb}(Chebotar\"ev) Suppose
  that $\om \in \cc$ is a primitive $p^{th}$ root of unity where $p$
  is a prime. Let $V$ be the Fourier matrix with $(i, j)$-entry equal
  to $\om^{ij}$ , for $0\leq i,j \leq p-1$.  Then all square submatrices
  of $V$ have nonzero determinant.
\end{theorem}

Let $F_n$ denote the Fourier $n\ti n $ matrix over $\cc$ with $F_nF_n^*=I_n$. 
Here $nF^*$ is now the complex conjugate transposed of $F_n$.

 We can define unit-derived codes using the unit $F_n$ (or the unit
 $nF_n^*$). Suppose then $\mathcal{C}_r$ is a unit-derived $(n,r)$ code
 where $\mathcal{C}_r$ is defined using any $r$ rows of $F_n$ and the
 check matrix may be obtained directly from $F_n^*$ as explained
 above; the check matrix may also be obtained directly from $nF_n^*$ and this
 is often more convenient.

\begin{theorem}\label{distancemds} Suppose $n$ is prime. Then the 
 distance of $\mathcal{C}_r$ is $(n-r+1)$.
\end{theorem}
\begin{proof} The proof follows from Theorems \ref{bigdist} and
  \ref{cheb}.

\end{proof}  

Thus any such $\mathcal{C}_r$ is an $(n,r,n-r+1)$ mds code when $n$ is
prime.  Any collection of $r$ rows of $F_n$ may be used to generate an
mds $(n,r,n-r+1)$ code. Hence there are ${n\choose r}$ mds
$(n,r,n-r+1)$ different codes derived from the single unit $F_n$.
 
Cyclic codes  using complete orthogonal sets of idempotents related to the
Fourier matrix are obtained in section \ref{four}. These will make it
easier to derive systems of mds codes over $\R$

Sections \ref{finite} deals with the construction of
series of mds codes over finite fields and later in section \ref{four}
series of cyclic such codes are constructed.

\subsubsection{Example} Let $\om$ be a primitive $7^{th}$ root of
$1$ in $\cc$. Consider $F_7 = \begin{pmatrix} 1&1&1&1&1&1&1 \\ 1 & \om
  &\om^2&\om^3&\om^4&\om^5&\om^6\\ 1&\om^2&\om^4&\om^6&\om&\om^3&\om^5
  \\ 1&\om^3&\om^6&\om^2&\om^5&\om&\om^4
  \\ 1&\om^4&\om&\om^5&\om^2&\om^6&\om^3\\ 1&\om^5&\om^3&\om&\om^6&\om^4&\om^2\\ 1
  &\om^6&\om^5&\om^4&\om^3&\om^2&\om \end{pmatrix}$.
 
Let $\mathcal{C}_4$ be the code generated by the following matrix:
$A=\begin{pmatrix} 1&\om&\om^2&\om^3&\om^4&\om^5&\om^6
\\ 1&\om^2&\om^4&\om^6&\om&\om^3&\om^5
\\ 1&\om^5&\om^3&\om&\om^6&\om^4&\om^2
\\ 1&\om^6&\om^5&\om^4&\om^3&\om^2&\om \end{pmatrix}$.

$A$ has $\rank 4$.  A check matrix 
%(that is a matrix $D$ of $\rank 3$ such that $AD\T=0$) 
for $\mathcal{C}_4$ is $\begin{pmatrix} 1&1&1&1&1&1&1 \\ 1&
  \om^4&\om&\om^5&\om^2&\om^6&\om^3
  \\ 1&\om^3&\om^6&\om^2&\om^5&\om&\om^4 \end{pmatrix} $. This has
$\rank 3$. 

The code $\mathcal{C}_4$ is a $(7,4,4)$ code. Indeed ${7\choose 4}=35$
different such codes may be derived from $F_7$. 

\subsection{Finite fields}\label{finite}
In the finite field case it is not true in all cases when $n$ is prime
that the Fourier matrix $F_n$, when it exists, has non-zero
determinant of each square submatrix.  The purpose now is find finite
fields $K$ and primes $p$ such the Fourier $F_p$ matrix over $K$ has
non-zero determinant of each square submatrix.
In order that the Fourier $p\ti p$ matrix over $K$ should exist,
it is necessary that ${\rm{char}}\, K \not | \, p$, and $p/(q-1)$
where $q$ is the order of the field $K$.   

Say a square matrix $M$ over the field $K$ has the {\em Chebotar\"ev
  property} if the determinant of any square submatrix is non-zero.
By \cite{isaacs1} if the characteristic of $K$ is $0$, the Fourier
matrix $F_n$ over $K$ has the Chebotar\"ev property for a prime $n$.

See the paper \cite{isaacs1} for details on the following.
%The notation in this paper is adhered to.
$F[G]$ denotes the group ring of the group $G$ over the field $F$.
Let $z$ be a generator for the cyclic group $G$ of order a prime
$p$. Each vector $v \in F[G]$ is uniquely in the form $f(z)$, where
$f\in F[X]$ and $\deg f < p$. The quantity $t=t(v)$ which is
$|\supp(v)|$ is exactly the number of non-zero coefficients in the
polynomial $f$ and this number is written as $t(f)$.

%The quantity $t = t(v)$ is defined as the number of nonzero
%coefficients in the polynomial $f$, and we write $t(f)$ to denote this
%number. Also 
Now $d(v)$ denotes the dimension of the space generated by $v$.
 
As shown in \cite{isaacs1} if $K$ is a field containing a primitive
$p^{th}$ root of unity, then the conclusion of Chebotar\"ev's theorem
over $K$ is equivalent to the assertion that $t(v)+d(v) > p$ for all choices
of nonzero vectors $v \in K[X]$. %The definitions of $t,d$ are given
%The paper \cite{isaacs1} may be consulted for details.  
In  \cite{isaacs1} cases  %where the
%conclusion of Chebotar\"ev's theorem fails. % to be constructed.
%As noted not every Fourier matrix
%$F_p$, when it exists, over a field $K$ of non-zero characteristic 
%has the Chebotar\"ev property. 
 of finite fields and primes $p$ with $t(v) + d(v) \leq p$ were found and
  %This enabled the authors of \cite{isaacs1} to find examples where
 %Chebotar\"ev's theorem fails in prime characteristic. 
 this enabled the authors of \cite{isaacs1} to find examples where
 Chebotar\"ev's theorem fails in prime characteristic. 
%hence the Chebotar\"ev property does not hold in these cases 
%for the Fourier matrices. % over these fields. 
The following theorem of \cite{isaacs1} gives necessary and
sufficient conditions for this failure to occur, where the conditions
are expressed in terms of the polynomial ring $K[X]$.  
%Let $z$ be a
%generator for $G$ the cyclic group of order $n$ and note that each
%vector $v \in K[G]$ can be uniquely written in the form $f(z)$, where
%$f \in K[X]$ and $\deg f < p$.

\begin{theorem}\label{isaacs3}(Goldstein, Guralnick, Isaacs
  \cite{isaacs1}, (6.3) Theorem).  Let $G = \langle z \rangle$ be a
  group of prime order $p$ and suppose that $v \in K[G]$ is nonzero,
  where $K$ is an arbitrary field. Write $v = f(z)$, where $f \in
  K[X]$ and $\deg f < p$. Then $t(v) + d(v) \leq p$ if and only if
  $t(f) \leq \deg h$, where $h(X) = \gcd(X^p - 1, f(X))$.
\end{theorem}

It is worth noting that the examples given in the paper
\cite{isaacs1}
(pages 4035-6) for
 which Chebotar\"ev Theorem fails use (distinct) primes $p,q$
 in which  the order of $q \mod p$ is less than
 $\phi(p)= p-1$;  this should be compared with Theorem \ref{nicely}
 below. 
We are interested in finite fields $K$ and primes $p$ for which the
Fourier matrix over $K$ exists and satisfies the Chebotar\"ev
condition.

The paper \cite{isaacs1} argues as follows to show that for each prime
$p$, there are only finitely many characteristics where Chebotar\"ev
can fail: % and thus there are only finitely many characteristics
%where
%examples such as those discussed above can occur: 
``Consider the determinants of all square submatrices of the complex
matrix $[\zeta_{ij}]$, as in Theorem \ref{cheb}.  These are algebraic
integers, and they are nonzero by Chebotar\"ev's theorem, and so their
norms are nonzero rational integers. It should be reasonably clear
that the characteristics where the conclusion of Chebotar\"ev's can
fail are exactly the primes that divide at least one of these
integers, and clearly, there are just finitely many such primes.''
\subsection{Fourier matrix over finite fields}\label{four}
In order to construct the Fourier matrix $F_p$ over $GF(q)$ it is
necessary that $p /(q-1)$. For given unequal primes $p,t$ by Fermat's
little Theorem $p/(t^{\phi(p)}-1)$. As $p$ is prime, $\phi(p) =p-1$.
%Now $p-1$ may not be the smallest $r$ such
%that $p/(t^r-1)$ and we can look for such amongst the divisors of $p-1$.
For given unequal primes $p,t$ there is a field $GF(t^r)$ such that
$p/(t^r-1)$ and the Fourier matrix $F_p$ exists over this field.

Let $p,q$ be unequal primes and $K=GF(q^{\phi(p)})$. Then
$p/(q^{\phi(p)}-1)$ and the Fourier matrix $F_p$ exists over $K$.

\begin{lemma}\label{primes} Let $p,q$ be unequal primes. 
Suppose the order of $q \mod p$ is $ \phi (p)$. Then $(x^{p-1}+ x^{p-2}+
\ldots + x+1)$ is irreducible over $GF(q)$.
\end{lemma}
\begin{proof} It is known that the cyclotomic polynomial $\Phi_n(x)$ 
factors over a finite field $GF(q)$ into irreducible polynomials of
degree $r$ where $r$ is the order of $q \mod n$. Here $\Phi_p(x) =
x^{p-1}+ x^{p-2} + \ldots + x +1 $ and $r= \phi(p)=p-1=\deg(\Phi_p(x))$
and so $\Phi_p(x)$ is irreducible.
\end{proof}

\begin{theorem}\label{nicely}
 Let $p,q$ be unequal primes and $K=GF(q^{\phi(p)})$.
%Then
%$p/(q^{\phi(p)}-1)$ and the Fourier matrix $F_p$ exists over $K$.
 Suppose the order of $q \mod p$ is $\phi(p)$
 and (hence) that $f(x)=(x^{p-1} + x^{p-2}
 + \ldots + x + 1)$ is irreducible over $GF(q)=\mathbb{Z}_q$. Then the
 Fourier matrix $F_p$ exists over $K$ and satisfies the Chebotar\"ev
 condition.
\end{theorem}
\begin{proof} It has already been noted that $F_p$ exists.

Now $GF(q^{\phi(p)}) \cong GF(q)[\al] \cong \mathbb{Z}_q[\al] \cong
\frac{\mathbb{Z}_p[x]}{(\langle f(x)\rangle)}$ where $\al$ is the
cofactor $x + \langle f(x) \rangle$.

For $\om$ a primitive $p^{th}$ root of $1$ in $\cc$, $\mathbb{Z}[\om]
\cong \frac{\mathbb{Z}[y]}{\langle f(y) \rangle }$. This gives the
natural map $\mathbb{Z}[\om]\cong \frac{\mathbb{Z}[y]}{\langle
  f(y)\rangle} \rightarrow \frac{\mathbb{Z}_p[x]}{\langle
  f(x)\rangle}=GF(q)[\al]$. The kernel of this map are polynomials of
degree less than $p$ in $y$ in which each coefficient is divisible by
$p$.

This mapping may be extended $ \frac{\mathbb{Z}[y]}{\langle
  f(y)\rangle}[z] \rightarrow \frac{\mathbb{Z}_q[x]}{\langle
  f(x)\rangle}[z]$.

Suppose now $g(z) \in \frac{\mathbb{Z}_q[x]}{\langle f(x)\rangle}[z]$
satisfies $\deg g < p$ and let $h(z)=\gcd(g(z),z^p-1)$. Consider then
$\hat{g}(z) \in \frac{\mathbb{Z}[y]}{\langle f(y)\rangle}[z]\subset
\frac{\mathbb{Q}[y]}{\langle f(y)\rangle}[z]$ with the pre-image of
the coefficients of $g(z)$ as the coefficients of $\hat{g}(z)$.  Let
then (in $\frac{\mathbb{Z}[y]}{\langle f(y)\rangle}[z]$)
$\hat{h}=\gcd(\hat{g}(z),z^p-1)$. Now by Theorem \ref{isaacs3}
$t(\hat{g}) > \deg \hat{h}$.

Let $\frac{\Z[y]}{\langle f(y)\rangle} = \Z[\om]$ where $\om$ is a
primitive $p^{th}$ root of $1$ (in $\cc$) and $\frac{\Z_q[x]}{\langle
  f(x)\rangle} = \Z_q[\al]$ where $\al$ is a primitive $p^{th}$ root
of $1$ in $\Z_q$.
 
Now in $z^p-1 = \prod_{i=0}^{p-1}(z-\om^i)$ in $\Z[\om]$ and $z^p-1 =
\prod_{i=0}^{p-1}(z-\al^i)$ in $GF(q^{p-1})=\Z_q[\al]$.

Thus in $\Z_p[\al]$, $\gcd(g(z), z^{p}-1) = \prod_{j\in J}(z-\al_j)=
h(z)$ where $J$ is a proper subset of $I=\{0,1,\ldots, p-1\}$.  In
$\Z[\om]$, $\gcd(\hat{g}(z), z^p-1) = \prod_{j\in
  J}(z-\om_j)=\hat{h}(z)$.

Hence $\deg \hat{h}(z)= \deg h(z)$.  Thus it is seen that since
$t(\hat{g}(z)) = t(g(z)$ and $\deg \hat{h}(z)= \deg h(z)$ and
$t(\hat{g}) > \deg \hat{h}$ that $t({g}(z)) > \deg h(z)$. Hence by
Theorem \ref{isaacs3} the Fourier matrix $F_p$ over $GF(q^{\phi(p)})$
satisfies Chebotar\"ev's condition.
\end{proof}

Thus fields $GF(q^{\phi(p)})$ with $p,q$ unequal primes where
the order of $q$ mod $p$ is $\phi(p)= p-1$, and (hence) where $x^{p-1}+
x^{p-2}+ \ldots + 1$ is irreducible over $GF(q)$ is such that  the Fourier
matrix $F_p$ over $GF(q^{\phi(p)})$ satisfies the Chebotar\"ev
property. There are clearly many such examples and particular ones 
are given in section \ref{Examples}.

%\subsection{Germain}
\subsection{Germain type}
A prime $p$ is a {\em Germain prime} if $(2p+1)$ is also a prime. A {\em
  safe prime} is one of the form $(2p+1)$ where $p$ is prime. 

\begin{proposition}\label{ger} 
Suppose $p$ and $q=(2p+1)$ are primes. Then the Fourier matrix $F_p$
exists over $GF(q)$ and satisfies the Chebotar\"ev condition.
\end{proposition}
\begin{proof} 
Now $p/(q-1)$  and 
the order of $q \mod p$ is $1$. Let $\al$ be an element of order $2p=q-1$
in $GF(q)$. Then $\al^2$ has order $p$ and the Fourier matrix $F_p$ over $GF(q)$
then exists and can be constructed from powers of $\al^2$.   Let $f(x)$
be a polynomial of degree less than $p$ and consider $\gcd((x^p-1),
f(x)) = h(x)$ in $GF(q)$. Now in $GF(q)$, $ x^p-1 =
\prod_{i=0}^{p-1}(x-\al^{2i})$ as each $\al^{2i}, 0\leq i \leq (p-1)$
is a root of $x^p-1$. Hence $h(x) = \gcd((x^p-1),f(x)) = \prod_{j\in
  J}(x-\al^{2j}$ where $J\subseteq \{0,1,\ldots, p-1\}$. Let
$\om$ be a primitive $p^{th}$ root of $1$. Consider $(x^p-1), f(x)$ as
polynomials in $\Z[x]$. Now $t(f)$ in $GF(q)$ is the same as $t(f)$
in $\Z[x]$. Then $\gcd((x^p-1), f(x)) = h(x)$ satisfies
$t(f)> \deg(h(x)$ as elements in $\Z[x]$. 
Now $h(x) = \prod_{j\in \hat{J}} (x-\om^j)$ for
$\hat{J}\subseteq \{0,1,\ldots, (p-1)\}$. Now $\hat{J}=J$ and so $\deg
h(x)$ in $\cc[x]$ must be the same as $\deg(h(x)$ in $GF(q)[x]$. Hence
the Fourier matrix $F_p$ over $GF(q)$ satisfies the Chebotar\"ev
condition.    
\end{proof}

%This argument works $q=np+1$ is a prime for a prime $p$. 

The Fourier matrices in these cases are particularly nice as they
consist of integers modulo a prime $q$.

\subsection{Examples}\label{Examples}

A Computer Algebra system such as 
 GAP \cite{gap},  MAPLE or MATLAB is  useful
for calculations. % and applications.

A circulant matrix is a matrix of the form $\begin{pmatrix}a_0 & a_1 &
  \ldots & a_{n-1} \\ a_{n-1} & a_ 0 & \ldots & a_{n-2} \\ \vdots &
  \vdots & \vdots & \vdots \\ a_1& a_2 &\ldots &
  a_0 \end{pmatrix}$. Thus $\cir(a_0,a_1,\ldots, a_{n-1})$ will denote
the circulant matrix with first row $(a_0,a_1,\ldots, a_{n-1})$.

\subsubsection{$GF(2^r)$}

\begin{enumerate}
\item $GF(2^2)$: The order of $ 2 \mod 3$ is $2$ and 
  $(x^2+x+1)$ is irreducible over $GF(2)$. Thus $F_3 = \begin{pmatrix} 1 &
  1 & 1 \\ 1 &\om & \om^2 \\ 1 & \om^2 & \om \end{pmatrix} $ has the
  Chebotar\"ev property where $\om$ is a primitive $3rd$ root of unity
  in $GF(4)$. This gives ${3\choose 2}= 3$ codes of type $(3,2,2)$.
\item $GF(2^4)$. The order of $2 \mod 5$ is $4$ and $(x^4+x^3+x^2+x+1)$
  is irreducible over $GF(2)$. Hence by Theorem \ref{nicely} the
  Fourier matrix $F_5$ exists over $GF(2^4)$ and satisfies
  Chebotar\"ev's condition that every square submatrix has
  determinant non-zero. Consider $F_5$.

Let $\al$ be a primitive element and define $\om = \al^3$.  Then
$F_5= \begin{pmatrix}1&1&1&1&1 \\ 1&\om & \om^2 & \om^3 & \om ^ 4 \\ 1
  & \om^2& \om^4 &\om & \om ^3 \\ 1 & \om^3&\om & \om^4&\om \\ 1 &
  \om^4 & \om^3&\om^2&\om \end{pmatrix}$ has the determinant of every
submatrix non-zero. We can use $F_5$ to define maximal distance
separable (mds) codes over $GF(16)$. So for example choosing $3$ of
the rows to get a generator matrix and then use the other two
corresponding of $F^*$ as check matrix gives $(5,3,3)$ codes. In total
this gives ${5\choose 3} = 10$ different $(5,3,3)$ codes.

\item $GF(2^6)$: Now $7/(2^6-1)$ and so the Fourier $F_7$ exists over
  $GF(2^6)$. However $(x^3+x+1)$, which is `missing a term', is a factor
  of $(x^7-1)$ and so $F_7$ does not satisfy Chebotar\"ev's
  condition. Here the order of $2 \mod 7$ is $3$ and $(x^3+x+1)$ is
  irreducible over $GF(2)$.
%To get cyclic codes consider $E_0 = \circ(1,1,1,1,1), E_1=\circ(1,\om,
%\om^2,\om^3,\om^4), E_2=\circ(1,\om^2,\om^4,\om,\om^3),
%E_3=\circ(1,\om^3,\om,\om^4,\om^2),
%E_4=\circ(1,\om^4,\om^3,\om^2,\om)$. Then choose 3 of
%$\{E_0,E_1,E_2,E_3,E_4\}$ as the generator matrix and the other 2 as
%the check matrix. 
\item $GF(2^{10})$. The order of $2 \mod 11$ is $\phi(11)=10$ and
  $(x^{10}+x^9 + \ldots + x +1)$ is irreducible over $GF(2)$. Thus by
  Theorem \ref{nicely} the Fourier $F_{11}$ over $GF(2^{10})$ has the
  Chebotar\"ev property and mds codes may be constructed from it. For
  example %${11 \choose 6}= 462$ different mds $(11,6,6)$
 % codes and 
 ${11 \choose 7} = 330$ mds 
  $(11,7,5)$ codes (of rate $\frac{7}{11}$) may be constructed over
  $GF(2^{10})$  and each of
  these is $2$-error correcting.
\item $GF(2^{12})$: The order of $2 \mod 13$ is $\phi(13)=12$ so by lemma
  \ref{primes} $(x^{12}+x^{11}+ \ldots + 1)$ is irreducible over
  $GF(2)$. Thus by Theorem \ref{nicely} $F_{13}$ over $GF(2^{12})$
  exists and satisfies Chebotar\"ev's condition. So for example this
  enables the construction of ${13 \choose 7} = 1716$ (different)
  codes of type $(13,7,7)$ in $GF(2^{12})$ which are then $3$-error
  correcting. 
%Further for example
 % ${13\choose 9}=715$ rate $\frac{9}{13}$ codes of form $(13,9,5)$ may be
 % constructed from $F_{13}$ over $GF(2^{12})$.

\vdots

\end{enumerate}

 \subsubsection{$GF(3^r)$:}
\begin{enumerate}
\item $GF(3^4)$: The order of $3 \mod 5$ is
  $\phi(5)=4$ and so the polynomial $(x^4+x^3+x^2+x+1)$ is irreducible over
  $GF(3)$. The Fourier matrix $F_5$ over $GF(3^4)$ exists and 
  has the Chebotar\"ev property by Theorem \ref{nicely} from which mds
  codes can be constructed.
\item $GF(3^6)$: The order of $3 \mod 7$ is $6$ and 
%$3^6-1= 2^3\ti 7  \ti 13$.  
$(x^6+x^5+ \ldots + x+1)$ is irreducible over $GF(3)$. Hence
  by Theorem \ref{nicely} $F_7$ exists and satisfies Chebotar\"ev's
  condition. This enables the construction of mds codes from
  $F_7$. For example ${7\choose 3}=35$ mds $(7,3,5)$ codes may be
  formed in $GF(3^6)$.
\item $GF(3^{16})$: The order of $3 \mod 17$ is $16$ and $(x^{16}+
  x^{15} + \ldots + x+ 1)$ is irreducible over $GF(3)$. Hence by
  Theorem \ref{nicely} $F_{17}$ satisfies Chebotar\"ev's
  condition. This enables the formation of mds codes from
  $F_{17}$. For example ${17 \choose 9}= 24310 $ mds codes $(17,9,
  9)$ and ${17 \choose 13}= 2380$ mds codes $(17,13,5)$ may be
  constructed from $F_{17}$ in $GF(3^{16})$.
\end{enumerate}

\subsubsection{$GF(5^r)$:} \begin{enumerate}
\item $GF(5^2)$: The order of $5 \mod 3$ is $2$ and $(x^2+x+1)$ is
  irreducible in $GF(5)$. Thus the Fourier $F_3$ exists in
  $GF(5^2)$ and has Chebotar\"ev property.
\item $GF(5^6)$: The order of $5 \mod 7$ is $6$ and
  $(x^6+x^5+x^4+x^3+x^2+x+1)$ is irreducible in $GF(5)$. Thus the
  Fourier matrix $F_7$ over $GF(5^6)$ exists and satisfies
  Chebotar\"ev's property. Hence for example it may be used to
  construct ${7\choose 4} = 35$ different mds $(7,4,4)$ codes over
  $GF(5^6)$ and indeed ${7\choose 5}= 21$ different $(7,5,3)$ codes
  over $GF(5^6)$.

 \end{enumerate}

\subsubsection{$GF(7^r)$}
\begin{enumerate}
%\item $GF(7)$: The polynomial $x^3-1$ has irreducible factors
 %     $(x-1),(x+\al), (x+\al^5)$ where $\al$ is a primitive root.
%It may still be verified that $F_3$ satisfies Chebotar\"ev's condition
 %     using Theorem \ref{isaacs3}.
\item $GF(7^4)$: The order of $7 \mod 5$ is $4$ and $(x^4+x^3+x^2+x+1)$
  is irreducible over $GF(7)$. Hence by Theorem \ref{nicely}, $F_5$
  exists over $GF(7^4)$ and satisfies Chebotar\"ev's condition. Hence
  mds codes may be constructed from $F_5$.
\item $GF(7^{10})$: The order of $7 \mod 11$ is $10$ and
  $(x^{10}+x^9+\ldots + x+ 1)$ is irreducible over $GF(7)$.  Thus by
  Theorem \ref{nicely} $F_{11}$ exists over $GF(7^{10}$ and satisfies
  Chebotar\"ev's condition.
\end{enumerate}

\subsubsection{$GF(11^r)$} 
\begin{enumerate}
\item 
$GF(11)$: Here $5/(11-1)$ and so the Fourier matrix $F_5$ exists over
  $GF(11)$.  Theorem \ref{nicely} cannot be applied as the
  irreducible factors of $(x^5-1)$ in $GF(11)$ are $\{x-1,x-\al^2,
  x-\al^4,x-\al^6,x-\al^{8}\}$, where $\al$ is a primitive element in
  $GF(11)$. (This $\al$ can be chosen to be $2$ as the order of $2
  \mod 11$ is $10$.)  However $5$ is a
  Germain prime (with safe prime $11=5\ti 2 +1$ and so the Proposition \ref{ger}
  may be applied.  
%For any $g(x)$ it is seen that
 % $h(x)=\gcd(x^5-1, g(x)) = \prod_{j\in J}(x-\al^{2j})$ where $
  %J\subseteq \{0,1,2,3,4\}$. Then it follows as in the proof of
  %Theorem \ref{nicely} that $t(g)> \deg h$.
%It can then be verified that if $\deg f(x) <5$ 
%$x^5-1$ has $t(f) > \deg f$ where $t(f)$ is as in the notation
%of \cite{isaacs1} (that is, $f(x)$ is not `missing a term') and that in
%general for any multiple $f(x)$ of a divisor of $x^5-1$ in which $\deg
%f(x) < 5$ that $t(f)> \deg h(x)$ where $h(x)=\gcd(x^5-1,f(x))$.

Thus the Fourier $F_5$ over $GF(11)$ has the Chebotar\"ev
property. From this mds codes many be constructed.  Here $2$ is a
primitive root and so $2^2=4$ has order $5$.
% Note $5^{-1} = 9\mod $ and 
%$3^2=9$ so that $1/\sqrt{5}= 3$. 
Thus then

$F_5 = \begin{pmatrix}1&1&1&1&1 \\ 1 & 4 & 4^2 &4^3&4^4 \\ 1 & 4^2
  &4^4 & 4&4^3 \\ 1 & 4^3 & 4& 4^4&4^2 \\ 1 & 4^4&4^3&4^2&4
       \end{pmatrix}=  \begin{pmatrix} 1&1&1&1&1 \\ 1&4 & 5 & 9& 3 \\ 1 & 5 & 3&4& 9 \\
	1& 9&4&3&5 \\ 1&3&9&5&4\end{pmatrix}$

is a Fourier matrix over $GF(11)$ which has the Chebotar\"ev
property. %Here $F_5F_5^*=5I_5$ where $F^*=\hat{F\T}=\hat{F}$.
%This then is 
%$F_5 = $.
This gives for example ${5\choose 3} = 10$ mds codes $(5,3,3)$ over
$\mathbf{Z}_{11}$ which are $1$-error correcting.
\item $GF(23)$. Here $p=11$ is a Germain prime with safe prime
  $q=2p+1=23$. The Fourier matrix $F_{11}$ exists over $GF(23)$
  exists and by Proposition \ref{ger} it satisfies the Chebotar\"ev
  condition. In $GF(23)$ a primitive element is $5$ and so $5^2=2$ is
  an element of order $11$ from which the Fourier matrix $F_{11}$ over
  $GF(23)$ can be constructed. This gives

$F_{11} =\begin{pmatrix} 1 & 1 &1 & \ldots &1 \\ 1 & 2 & 2^2 & \ldots
  & 2^{10} \\ 1 & 2^2& 2^4 &\ldots & 2^{20} \\ \vdots & \vdots &
  \vdots &\vdots &\vdots \\ 1 & 2^{10}& 2^{20} & \ldots
  2^{100} \end{pmatrix} =\begin{pmatrix} 1 & 1 &1 & \ldots &1 \\ 1 & 2
  & 4 & \ldots & 12 \\ 1 & 4& 14 &\ldots & 6 \\ \vdots & \vdots &
  \vdots &\vdots &\vdots \\ 1 & 12& 6 & \ldots & 2 \end{pmatrix}$.

%Here $F_{11}F_{11}^*= 11 I_{11}$.

\item $GF(11^3)$: Now $7/(11^3-1)$ and $(x^7-1)$ has irreducible factors
  $(x-1), (x^3+\al^4 x^2+\al^2 x - 1), (x^3+\al^7x^2+\al^9x -1 )$ over
  $GF(11)$ where $\al$ is primitive. As pointed out in \cite{isaacs1}, 
  $F_7$ over $GF(11^3)$ does not have the the Chebotar\"ev property.

\item A large example: Consider $GF(227)$. The Fourier matrix
  $F_{113}$ exists over $GF(227)$ and by Proposition \ref{ger}
  satisfies the Chebotar\"ev property since $113$ is a Germain prime
  with matching safe prime $227$.   This for example enables the
  construction of ${113 \choose 57}$ (different) mds $(113,57,57)$ codes over
  $\Z_{227}$. The number ${113 \choose 57}$ is of order $10^{32}$ or
  $2^{109}$. Also for example ${113 \choose 99}$ high rate mds
  $(113,99,15)$ codes may be constructed over $\Z_{227}$ which are
  $7$-error correcting. This number
  ${113 \choose 99}$ is of order $10^{17}$ or $2^{57}$. 
\end{enumerate}

%Similar cyclic codes may also be formed as explained below in section \ref{above}.

\section{Codes from complete orthogonal sets of idempotents}\label{above}
\subsection{Notation}
Let $R$ be a ring with identity $1_R =1$. In general $1$ will denote
the identity of the system under consideration. A {\em complete family
  of orthogonal idempotents} is a set $\{e_1, e_2, \ldots, e_k\}$ in
$R$ such that \\ (i) $e_i \not = 0$ and $e_i^2 = e_i$, $1\leq i\leq
k$;\\ (ii) If $i\not = j$ then $e_ie_j = 0$; \\ (iii) $1 = e_1+e_2 +
\ldots + e_k$.

The idempotent $e_i$ is said to be {\em primitive} if it cannot be
written as $e_i= e_i^{'}+ e_i^{''}$ where $e_i^{'},e_i^{''}$ are
idempotents such that $e_i^{'},e_i^{''} \neq 0$ and
$e_i^{'}e_i^{''}=0$. A set of idempotents is said to be {\em
  primitive} if each idempotent in the set is primitive.

Methods for constructing complete orthogonal sets of idempotents are
derived in \cite{hur3}.  Such sets always exist in $FG$, the group
ring over a field $F$, when $char F \not | \, |G|$. See \cite{seh} for
properties of group rings and related definitions. These idempotent
sets are related to the representation theory of $FG$.  Other methods
for constructing complete orthogonal sets of matrices 
such as from orthonormal bases are considered in \cite{hur5}.

%A mapping $^*: R\to R$ in which $r\mapsto r^*, (r\in R)$ is said to be
%an {\em involution} on $R$ if and only if (i) $r^{**} = r, \, \forall
%r\in R$, (ii) $(a+b)^* = a^*+b^*, \, \forall a,b \in R$, and (iii)
%$(ab)^* = b^*a^* , \, \forall a,b \in R$.  We are particularly
%interested in the case where $^*$ denotes complex conjugate transpose
%in the case of matrices over $\cc$ and denotes transpose for matrices
%over other fields. An element $r\in R$ is said to be {\em symmetric}
%(relative to 
%$^*$) if $r^* = r$ and a set of elements is said to be symmetric if
%each element in the set is symmetric.

%\input{rankidem}
\subsection{Rank}\label{grmat1}

\begin{lemma}\label{trrank} Suppose $\{E_1,E_2, \ldots, E_s\}$ is a
set of orthogonal idempotent matrices. Then $\rank (E_1+E_2 +\ldots +
E_s) = \tr (E_1+E_2+ \ldots + E_s) = \tr E_1+ \tr E_2+ \ldots + \tr
E_s = \rank E_1+ \rank E_2 + \ldots +\rank E_s$.
% or more generally $\rank (E_{i_1} + E_{i_2} + \ldots E_{i_s}) =
%\rank E_{i_1} +\rank E_{i_2}+ \ldots + \rank E_{i_s}$ 
\end{lemma}
\begin{proof}
It is known that $\rank A = \tr A$ for an idempotent matrix, see for
example \cite{idemrank}, and so $\rank E_i = \tr E_i$ for each $i$. 
If
$\{E,F,G\}$ is a set of orthogonal idempotent matrices so is
$\{E+F,G\}$. 
From this it follows (by induction) that $\rank (E_1+E_2 +\ldots + E_s)
= \tr (E_1+E_2+ \ldots E_s)= \tr E_1+\tr E_2 + \ldots + \tr E_s =
\rank E_1+ \rank E_2 + \ldots \rank E_s$.
\end{proof}
\begin{corollary}\label{trrank1}
$\rank(E_{i_1}+ E_{i_2}+ \ldots + E_{i_k})= \rank E_{i_1} +\rank
  E_{i_2}+ \ldots + \rank E_{i_k}$ for $i_j \in \{ 1,2,\ldots, s\}$,
  $i_j\neq i_l$.
\end{corollary}

\subsection{The codes}\label{above1}

Let $\{E_1, E_2, \ldots, E_k\}$ be a complete orthogonal set
of idempotents in $F_{n\ti n}$ and suppose $\rank E_i = r_i$ with then
$\sum_{i=1}^k r_i = n$. Let $I=\{1,2\ldots, k\}$ and suppose $J\subseteq
I$. Then by Lemma \ref{trrank} $\rank (\sum_{j\in J}E_j)= \sum_{j\in
  J} \rank (E_j)$.

Let $G=(E_1+E_2+ \ldots + E_s)$ with $s< k$ and $H=(E_{s+1}+\ldots +
E_k)$.  Let $r= \rank G = (r_1+r_2+\ldots+r_s)$, and then $(n-r)=\rank
H=(r_{s+1}+r_{s+2}+ \ldots +r_{k})=(n-r)$. Note that $GH=0$.

Let $\mathcal{C}_s$ denote the code with generator matrix $G$ and
check matrix $H\T$. Then $\mathcal{C}_s$ is an $(n,r)$ code.

\begin{lemma}\label{sum} Let $A\in F_{n\ti n}$. Then $AH=0$ if and only if $AE_i=0$
 for $i=s+1, s+2, \ldots, k$.
\end{lemma}
\begin{proof} Suppose $AH=0$. Multiply through on the right by $E_i$ for
 $s+1\leq i \leq k$. Then $AE_i=0$ as $E_iE_i=E_i$ and $E_iE_j=0$ for
  $i\neq j$. On the other hand if $AE_i$ for $i=s+1, s+2,\ldots,k$
  then clearly $AH=0$.
\end{proof}

%In a sense this means that a code can be decoded term by term.

Any $s$ elements of $\{E_1, E_2, \ldots, E_k\}$ can be used as a
generator matrix and then the other $(k-s)$ elements give the check
matrix. The ranks are determined by the ranks of the elements chosen.
Any complete orthogonal set of idempotents may be used and the reader
is referred to \cite{hur5} for general constructions of these. Here we
stick to cases related to the idempotents in the cyclic group
ring. 

Now suppose  $S=\{E_1, E_2, \ldots, E_n\}$ is a complete orthogonal set
of idempotents in $K_{n\ti n}$ where each $E_i$ has $\rank 1$. 
In this case it can be seen that choosing $r$
elements gives a $(n,r)$ code with the generator matrix given by the
sum of these $r$ elements and the check matrix given by (the transpose
of) the sum of the other $(n-r)$ elements. Each choice of the $r$
elements gives a different $(n,r)$ code so the set-up gives
$\binom{n}{r}$ different $(n,r)$ codes.

\subsection{Distances attained} Suppose now that $S=\{E_1,E_2, \ldots, E_n\}$ is
a complete %symmetric
 orthogonal set of idempotents in $F_{n\ti n}$ for
a field $F$. Let $F_n$ be the $n\ti n$ matrix consisting of the first
columns of each of $\{E_1,E_2, \ldots , E_n\}$.

Let $G$ be the matrix consisting of the sum of $r$ elements of $S$ and
let $H$ be the sum of the other $(n-r)$ elements of $S$. Then as
explained in section \ref{above1} this defines an $(n,r)$ code say
$\mathcal{C}_r$ with generator matrix $G$ and check matrix $H\T$.
\begin{theorem}\label{bigdist1}
 Suppose the determinant of any square submatrix of $F_n$ is
 non-zero. Then any such code $\mathcal{C}_r$ has distance $(n-r+1)$
 and is thus an $(n,r,n-r+1)$ mds code.
\end{theorem}

\begin{proof} Suppose $u=(u_1,u_2,\ldots, u_n) \in \mathcal{C}_r$ has support at most
  $(n-r)$. Thus $u$ has entry $0$ in $r$ places. Suppose $u$ has entry
  $0$ except (possibly) at places $\{u_{k_1}, u_{k_2},\ldots,
  u_{k_{n-r}}\}$. Define $\hat{u} = (u_{k_1},u_{k_2}, \ldots,
  u_{k_{n-r}})$.

Let $H=E_{j_1}+E_{j_2}+\ldots E_{j_{n-r}}$.  Now $uH=0$ and so by
Lemma \ref{sum}, $uE_{j_i}=0$ for $i=1,2,\ldots, (n-r)$.

Let the $k_l^{th}$ entry of the column of $E_{j_t}$ be denoted by
$E_{j_{t_l}}$. Then $\di\sum_{l=1}^{n-r}u_{k_l}E_{j_{i_l}} = 0$ for
$i=1,2,\ldots, (n-r)$. Let $T_i$ denote the column
$(E_{j_{i_1}},E_{j_{i_2}}, \ldots, E_{j_{i_{n-r}}})\T$.  Then this
says that $\hat{u}T_i = 0$ for $i=1,2,\ldots,(n-r)$. Hence
$\hat{u}(T_1,T_2,\ldots, T_{n-r}) = 0$.

Let $A$ be the $(n-r)\ti (n-r)$ matrix $(T_1,T_2,\ldots,
T_{n-r})$. This is a square submatrix of $F_n$ and so its determinant
is non-zero. Hence $\hat{u} = 0$ and so $u=0$.

\end{proof}

\subsection{Cyclic case}\label{7}

Let $N=\{E_0,E_1, \ldots, E_{n-1}\}$ be the primitive orthogonal
complete set of idempotents obtained from the cyclic group $C_n$ of
order $n$ in $\cc$.  Take $E_i=\cir(\om^i, \om^{2i},
\ldots \om^{(n-1)i})$ where $\om$ is a primitive $n^{th}$ root of $1$.

Let $\mathcal{C}_r$ be the code with generator matrix $G=(E_0+ E_1+
\ldots + E_{r-1})$ and check matrix (the transpose of) $H=E_r+E_{r+1}+ \ldots
+E_{n-1}$. Then $G$ has rank $r$ and $H$ has rank $(n-r)$ by Lemma
\ref{trrank} and so $\mathcal{C}_r$ is a $(n,r)$ code. The first $r$
rows of $G$ are independent and the first $(n-r)$ rows of $H$ are
independent by results in \cite{hur1,hur2}. Hence the first
$r$ rows of $G$ can be taken as the generator matrix. Similarly the
first $(n-r)$ rows of $H\T$ can be taken as the check matrix of
$\mathcal{C}_r$.

More generally choose the sum of any $r$ of $S=\{E_0,E_2,\ldots,
E_{n-1}\}$ to form a generator matrix $G_r$ of a code $C_r$ of size $(n,r)$
and the sum of the remaining $(n-r)$ elements give the matrix
$H_{n-r}$ where $H_{n-r}\T$ is a check matrix. As explained the first
$r$ rows of $G_r$ are linearly 
independent and these may be taken as the generator matrix of this
cyclic code and the first $(n-r)$ rows of $H_{n-r}\T$ may be taken as a
check matrix.
\begin{theorem}\label{distancemds2} Suppose $n$ is prime. Then the 
 distance of $\mathcal{C}_r$ is $(n-r+1)$.
\end{theorem}
\begin{proof} The proof follows from Theorems \ref{bigdist1} and
  \ref{cheb}.

\end{proof}  

The codes constructed in this section \ref{7} are  
cyclic codes and are also ideals in the group ring of the cyclic group. 
%The codes constructed from the units are not
%ideals and so although these codes are similar to those in Section
%\ref{8} they cannot be equal.

Note that using these orthogonal sets of idempotents it is then easy
to construct mds codes over $\mathbb{R}$ by combining complex
conjugate idempotents when constructing the generator matrix. This is
illustrated in the following examples. By using complete orthogonal
sets of idempotents in $\mathbb Q_{n\ti n}$, codes over $\mathbb Q$ may be obtained.

\subsubsection{Codes from idempotents, examples}

%a $(n,r)$ code. Now $w e_k = 0$ for $k>r$ giving $w (\sum_{r+1}^ne_j)=0$
%iff $we_j=0$ for $j>r$.

Consider from $\cc C_5$ the following complete orthogonal set of
idempotent giving $\{E_0,E_1,E_2,E_3,E_4\}$ with $E_0 =
\frac{1}{5}\cir(1,1,1,1,1), E_1= \frac{1}{5}\cir(1,\om,\om^2,\om^3,\om^4),
E_2=\frac{1}{5}\cir(1,\om^2,\om^4,\om,\om^3), \\ E_3=\frac{1}{5}
\cir(1,\om^3, \om, 
\om^4,\om^2), E_4=\frac{1}{5}\cir(1,\om^4,\om^3,\om^2,\om)$.

If we choose $U=(E_0+E_1+E_2)$ as a generator matrix of a code
$\mathcal{C}$ then $V=(E_3+E_4)$ gives the check matrix $V\T$ 
of $\mathcal{C}$.
By Theorem \ref{distancemds2} this code is a $(5,3,3)$ code.  The first
three rows of $U$ are linearly independent and constitute the
generator matrix. The first two columns of $V$ are
linearly independent and  any $2\times 2$
submatrix has $\det \neq 0$ which gives the distance $3$.
The generator matrix is  $U=(E_0+E_1+E_2) = \frac{1}{5}\cir(3, 1+\om + \om^2, 1+\om^2+\om^4,
1+\om^3+\om, 1+\om^2 +\om^3)$.

Suppose we wish to generate a real $(5,3,3)$ code from $\{E_0, E_1,
E_2,E_3,E_4\}$.  It is noted that $\{E_1, E_4\}$ and $\{E_2,E_3\}$ 
consist of pairs whose sums are real and that $E_0$ is real. Consider
$G=(E_0+E_1+E_4)$ as 
the generator matrix and $H=(E_2+E_3)$ as the transpose of the 
check matrix. Both $G$ and
$H$ are real and thus get a real $(5,3,3)$ code.

%(These codes may be equivalent to the ones constructed in other sections.?)   
\subsection{Over finite fields} We may now use Theorem \ref{bigdist1}
and analogies of Theorem \ref{nicely} and Proposition \ref{ger} 
to construct series of cyclic mds codes over
finite fields. Note that if $\{F_1,F_2,\ldots, F_k\}$ are cyclic
(circulant) orthogonal idempotent matrices and $\rank F_1+\rank F_2+
\ldots + \rank F_k = r$ then also $G=(F_1+F_2+\ldots + F_k)$ is circulant and
the first $r$ rows of $G$ are linearly independent; this follows for
example from \cite{hur1}. Thus generator and check matrices of the
$(n,r)$ codes produced are obtained from the $(n\ti n)$ matrices by
using the first $r$ rows for the check matrix of the (natural)
generator matrix and the first $(n-r)$ rows of the check matrix.

Consider then as in Section \ref{four} two unequal prime $p,q$ and
$GF(q^{\phi(p)})$ where the order of $q \mod p$ is $\phi(p)$ and
$x^{p-1} + x^{p-2}+ \ldots + x+1$ is irreducible over $GF(q)$. In
these cases by Theorem \ref{orthog} the Fourier matrix $F_p$ over
$GF(q^{\phi(p)})$ exists and satisfies Chebotar\"ev's condition.

\begin{theorem}\label{orthog} Let $p,q$ be unequal primes and
  $K=GF(q^{\phi(p)})$.
Suppose  the order of $q \mod p$ is $\phi(P)$ and (hence) that $(x^{p-1} +
  x^{p-2}+ \ldots + x +1) $ is irreducible over $GF(q)$. Let 
  $\om$ be a primitive $p^{th}$ root of $1$ in $K$. Define, (in
  $K_{p\ti p}$),  $E_i=\frac{1}{p}\cir(1,\om^i,\om^{2i}, \ldots,
  \om^{(p-1)i})$ for $i=0,1,\ldots, (p-1)$.  Then $S=\{E_0,E_1, \ldots,
  E_{p-1}\}$ is a complete orthogonal set of idempotents (each
  $E_i$ has $\rank 1$) and  the
  codes produced using any subset of $S$ are
  cyclic mds codes.
\end{theorem}
\begin{proof} It is easy to check that $S$ is a complete orthogonal
  set of idempotents in $GF(q^{\phi(p)})$. Then the first rows of the
  elements of $S$ constitute (a multiple of) the rows of the Fourier
  matrix $F_p$. The result then follows from Theorem \ref{bigdist1}.
\end{proof}   

\begin{proposition} Suppose $p$ and $q=2p+1$ are primes and that $\om$
  is a primitive $p^{th}$ root of $1$ in $K=GF(q)$. Define (in $K$)
 $E_i=\frac{1}{p}\cir(1,\om^i,\om^{2i}, \ldots,
  \om^{(p-1)i})$ for $i=0,1,\ldots, (p-1)$.  Then $S=\{E_0,E_1, \ldots,
  E_{p-1}\}$ is a complete orthogonal set of idempotents. Further the
  codes produced using any subset of $S$ are
  cyclic mds codes. 
\end{proposition}

The constructions are fairly general and examples are easy to
construct. 
Similar examples to those of section \ref{Examples} from the point
of view of orthogonal sets of idempotents may be derived. 
 A  small selection of corresponding
examples to those of \ref{Examples} are given below with details
    omitted.
\subsubsection{Examples in finite fields}
\begin{enumerate}
\item $GF(2^2)$: Let $\om $ be a primitive 3rd root of $1$ in
  $GF(2^2)$.  A complete orthogonal set of idempotents is
  $S=\{E_0=\cir(1,1,1), E_1=\cir(1,\om,\om^2),
  E_2=\cir(1,\om^2,\om)\}$.  The first rows of this set gives a
non-zero   multiple of the
  Fourier matrix $F_3$ which  has the Chebotar\"ev
  property. Thus choosing any subset of $S$ as a generator matrix
  determines an mds code and each such codes is cyclic.  This gives
  for example ${3\choose 2}= 3$ cyclic codes of type $(3,2,2)$.
% having
 % chosen the generator matrix $G=E_j+E_l$ the first two rows of $G$
  %are linearly independent and determine the code.
\item $GF(2^4)$. Let $\om$ be a primitive $5^{th}$ root of unity in
  $GF(2^4)$.  Consider the complete orthogonal set of idempotents $S=
  \{ E_0 = \cir(1,1,1,1,1), E_1=\cir(1,\om, \om^2,\om^3,\om^4),
  E_2=\cir(1,\om^2,\om^4,\om,\om^3),
  E_3=\cir(1,\om^3,\om,\om^4,\om^2),
  E_4=\cir(1,\om^4,\om^3,\om^2,\om)$. 
%$S=\{ E_i | E_i =\sum_{i=0}^4 \}
 % $.

The first rows of $\{E_0,E_1,E_2,E_3,E_4\}$ determine a non-zero 
multiple of the Fourier matrix
$F_5$ over $GF(2^4)$.

So for example choosing the sum of $3$ of the elements of $S$ gives
a $(5,3,3)$ code  
%the first three rows of the sum are linearly
%independent and form the generator $(3\ti 5)$ matrix. 
and this
gives ${5\choose 3} = 10$ different cyclic $(5,3,3)$ codes.
\item $GF(2^{10})$. Let $E_i=\cir(1,\om^i,\om^{2i},\ldots, \om^{10i})$
  where $\om$ is a primitive $11^{th}$ root of unity in $GF(2^{10})$
  and $S=\{E_0, E_1, \ldots, E_{10}\}$. Using the first rows of
  $E_0,E_1, \ldots, E_{10}$ constitutes the the Fourier $F_{11}$ over
  $GF(2^{10})$. Now by section \ref{four} this $F_{11}$ has the
  Chebotar\"ev property and hence codes formed
  using sums of elements from $S$ are mds codes which are also
  cyclic.
\item $GF(2^{12})$: Let $E_i=\cir(1,\om^i,\om^{2i},\ldots, \om^{12i})$
  where $\om$ is a primitive $13^{th}$ root of unity in $GF(2^{12})$
  and $S=\{E_0, E_1, \ldots, E_{12}\}$. Using the first rows of
  $E_0,E_1, \ldots, E_{12}$ constitutes a multiple of the Fourier $F_{13}$ over
  $GF(2^{12})$. Now by Section \ref{four} this $F_{13}$ has the
  Chebotar\"ev property and hence codes formed
  using sums of elements from $S$ are mds codes and these  are also
  cyclic. 

\vdots

%\end{enumerate}

% \subsection{$GF(3^r)$:}
%\begin{enumerate}
\item $GF(3^4)$: Construct $E_i=\cir(1, \om^i, \om^{2i}, \ldots,
  \om^{4i})$ where $\om$ is a primitive $5^{th}$ root of unity in
  $GF(3^4)$ and let $S=\{E_0,E_1,E_2,E_3,E_4\}$. Then the first rows
  of these constitute a multiple of the Fourier matrix $F_5$ over
  $GF(3^4)$ which has noted in section \ref{four} has the
  Chebotar\"ev property. Thus codes formed using subsets of 
$S$ are mds cyclic
  codes.
\item $GF(3^6)$: Construct $E_i=\cir(1, \om^i, \om^{2i}, \ldots,
  \om^{6i})$ where $\om$ is a primitive $7^{th}$ root of unity in
  $GF(3^7)$ and let $S=\{E_0,E_1,E_2,E_3,E_4,E_5,E_6\}$. The
  first rows of $\{E_0,E_1, \ldots, E_6\}$ constitute a multiple of 
the Fourier matrix
  $F_7$ over $GF(3^6)$ which has noted in section \ref{four} has the
  Chebotar\"ev property. Thus codes formed from subsets of $S$ are mds codes.
\item $GF(3^{16})$: Mds cyclic codes may be obtained from $\{E_0,
  E_1,\ldots, \ldots, E_{16}\}$; details are omitted. For example we
  may obtain ${17 \choose 9} = 24310 $ mds cyclic $(17,9,9)$ codes.

%  and ${17 \choose 14}= 680$ mds high rate cyclic $(17,14,4)$ codes
 % which are then 2-error correcting.
%\item 

\end{enumerate}

Further (cyclic) examples may be
  obtained similar to those in section \ref{Examples} using $GF(5^r)$, 
  $GF(7^r)$, $GF(11^r)$, and so on.

\subsection{Equality} 
The question arises in this case as to whether or not the codes produced
from idempotents in the group ring of the cyclic group are the same as
the (corresponding) unit-derived ones in section \ref{units} using
rows of the Fourier matrix. 
It may be shown that they have the same check matrix (the details are
omotted) and so they are
equal but this is not obvious from the way they are constructed and
going from one generator matrix to another is not easy. 
Each presentation has its own advantages. 

%\section{Decoding}
\section{Decoding}\label{decode}
 A minor variation of the Peterson-Gorenstein-Zierler algorithm, see
\cite{blahut} Chapter 6 for details, may be used for codes 
 where the 
chosen rows of the Fourier matrix as in section \ref{units} are
consecutive. The details are not included as  
better and more efficient decoding algorithms exist as shown below.
Cyclic codes may be decoded by any general technique for decoding
cyclic codes. 

Error-locating pairs and error-correcting pairs were introduced in 
\cite{koetter} and \cite{pell} and this is the approach taken here.

For vectors ${\bf{u}}=(u_0,u_1,\ldots,
u_{n-1})$ and ${\bf{v}}=(v_0,v_1,\ldots, v_{n-1}) $ define ${\bf{u}}*{\bf{v}}=
(u_0v_0,u_1v_1,\ldots,u_{n-1}v_{n-1}) $. For subspaces $U,V$ define
$U*V = \{{\bf{u}}*{\bf{v}} | {\bf{u}}\in U, {\bf{v}}\in V\}$.

Let $C^\perp$ denote the
orthogonal complement of $C$, $k(C)$ the dimension of $C$ and
$d(C)$ denote the (minimum) distance of $C$.
Let $U,V,C$ be linear codes over a field $K$. Say $(U,V)$ is a
{\em $t$-error locating pair} for $C$ if (i) $U*V \subseteq C^\perp$, (ii) $k(U)
> t$, (iii) $d(V^\perp) > t$.
If further (iv) $d(C)+d(U) >n$, where $n$ denotes the code length of
$C$, then say $(U,V)$ is a {\em $t$-error correcting pair} for $C$. 

We now show how $t$-error correcting pairs may  be
constructed for many of the $(n,r,n-r+1)$ codes as described by
the unit-derived method of section \ref{units} with $2t=n-r$. 

Suppose $K$ is a field which contains a primitive $n^{th}$ root of
unity $\om$ such that the inverse of $n$ exists in $K$.
Define $e_0= (1,1,\ldots, 1), e_1=(1,\om, \om^2, \ldots, \om^{n-1}),
\ldots, e_{n-1}=(1,\om^{n-1},\om^{2(n-1)}, \ldots,
\om^{(n-1)(n-1)})$. 
%Define $e_i\T=(1,\om^{-i},\om^{-2i},\ldots,\om^{-(n-1)i})$.

The set $S=\{e_0,e_1,\ldots,e_{n-1}\}$ is a basis for $K^n$ as
it consists of the rows of the Fourier matrix and so $S$ is a set of
$n$ linearly independent vectors in $K^n$.

The dot/scalar product of $u,v$
for vectors $u,v$ in $K^n$ is denoted by $u\cdot v$.

\begin{lemma}\label{jam} $e_i*e_j = e_{i+j}$ where $i+j$ is interpreted $\mod n$.
\end{lemma}
\begin{proof} $e_i*e_j = (1,\om^i,\om^{2i},\ldots,
  \om^{(n-1)i})*(1,\om^{j},\om^{2j}, \ldots, \om^{(n-1)j} =
  (1,\om^{i+j},\om^{2(i+j)}, \ldots, \om^{(n-1)(i+j)}) = e_{i+j}$.
\end{proof}

\begin{lemma}\label{jam1} Suppose $U=\langle {\bf {u_1,u_2,\ldots,u_{k}}}\rangle,
  V=\langle {\bf v_1,v_2,\ldots,v_s}\rangle$ for vectors ${\bf u_i,v_j}$.
Then $U*V \subseteq \langle {\bf u_i}*{\bf v_j} \, |\, 1\leq i\leq k, 1\leq j \leq
s \rangle$.
\end{lemma}

\begin{lemma}\label{jam2} Let $I=\{0,1,2,\ldots, n-1\}$ and $J\subseteq
    I$. Consider 
    $C =\langle e_j \, | \, j\in J\rangle$. Define $\hat{J}= \{ n-j
    \mod n \,
    | \, j\in J\}$ and  $K= (I - \hat{J})$. Then $C^\perp = \langle e_k \, | \,
    k\in K \rangle$. 
\end{lemma}
\begin{proof} This follows since $e_i\cdot e_j = 0$ if and only if
  $j=n-i \mod n$.
\end{proof} 

Suppose now the Fourier matrix with rows $e_i, 0\leq i \leq (n-1)$ 
has the Chebot\"arev
property that the determinant of any submatrix is non-zero. Then as
pointed out the code generated by any $r$ of the vectors $S$ is an
$(n,r,n-r+1)$ code. 

We now construct $t$-error correcting pairs for many of these
codes with maximum $t$. 

Suppose the $r$ vectors of $S$ are chosen
consecutively as $\{e_i,e_{i+1}, \ldots, e_{i+r-1}\}$ to form a code 
where suffices are
interpreted $\mod n$.
We shall show that in this case how to construct a (nice) $t$-error
correcting pair,  $2t=n-r$. 
We do this in the case of the code ${C}$ 
generated by $\{e_0,e_1,\ldots, e_{r-1}\}$; the other cases are
similar.  %Note that $e_i\dot e_j=0$ if and only if $j=n-i$
%and $e_0\dot e_0 \neq 0$. 
From Lemma \ref{jam2} it is seen 
 that $\langle e_1,e_{2}, \ldots, e_{n-r-1}\rangle
\subseteq C^\perp$. % as for $i=1,2,\ldots, n-r-1$, $e_{n-i}\not \in C$.
Set $U=\langle e_0,\ldots, e_{t}\rangle$. The dimension of $U$ is
$k(U)= t+1> t$ (as $\{e_0,e_1,\ldots, e_t\}$ is linearly independent). 
Set $V=\langle
e_1,\ldots, e_{t-1}\rangle$. Then $V^\perp = \langle e_0, e_1, \ldots,
e_{n-t-1}\rangle$. Now $V^\perp$ is a $(n,n-t, t+1)$ code and so
$d(V^\perp)>t$. Now by Lemma \ref{jam} and Lemma 
\ref{jam1}, $U*V\subseteq \langle
e_1,e_{2} ,\ldots, e_{2t-1}\rangle = \langle e_1,e_{r+1} \ldots,
e_{n-r-1} \rangle\subseteq C^\perp$. % as $2t=n-r$. 
Thus conditions (i),(ii),(iii) are satisfied for the pair $(U,V)$. Now
$U$ is a $(n,t+1,n-t)$ code and so $d(U)=n-t$. 
Hence $d(C)+d(U) = (n-r+1)+(n-t)= 2n-r-t+1= n+(n-r)-t+1 = n+2t-t+1 =
n+t+1 > n$. Thus condition (iv) is satisfied for the pair $(U,V)$ and
so $(U,V)$ is a $t$-error correcting pair.

Similarly it is also possible to construct (nice) $t$-error correcting pairs when the
$\{e_{i_1}, e_{i_2}, \ldots, e_{i_r}\}$ has other structures such as
when  
the consecutive differences $i_{j+1}-i_{j}$ are constant. The problem of getting
$t$-error correcting pairs  for a more general $\{e_{i_1}, e_{i_2}, \ldots,
e_{i_r}\}$ is left open. % for small $r$ they may be determined.

As an example consider the Fourier matrix $F_{11}$ over $K=GF(23)$
constructed in section \ref{Examples}. Use $\om=5^2=2$ which is a
primitive ${11}^{th}$ root of $1$ in $K$. Let the rows of $F_{11}$ be denoted
by $\{e_0,e_1,\ldots, e_{10}\}$ and  let $\mathcal{C}_7$ be the $(11,7,5)$
code generated by the first $7$ rows of $F_{11}$. 
We now define a $2$-error correcting pair $(U,V)$ as follows.
 
Define $U=\langle
e_0,e_1,e_2\rangle $ and $V=\langle{e_1,e_2}\rangle$. Then
(i) $U*V \subseteq \langle e_1,e_2,e_3,e_4\rangle \subseteq \mathcal{C}_7^\perp$,
(ii) $U$ has dimension $3$,
(iii) $V^\perp$ has distance $3$, (iv) $d(\mathcal{C}_7)+ d(U) = 5 + 9
> 11$. Thus $(U,V)$ is a $2$-error correcting  pair. The matrices $M(U),M(V)$ of
$U,V$ respectively are as follows:

$M(U)= \begin{pmatrix} 1& 1&1& \ldots  & 1 \\ 1 & \om & \om^2 & \ldots
& \om^{10} \\ 1& \om^2&\om^4 & \ldots & \om^{20} \end{pmatrix} =
\begin{pmatrix} 1& 1& 1& 1&1&1&1&1&1&1 & 1 \\ 1&2 &4 &
8&16&9&18&13&3&6& 12 \\ 
1 & 4 &
16 & 18&3&12&2&8&9&13 & 6 \end{pmatrix}$.

$M(V) = \begin{pmatrix}  1 & \om & \om^2 & \ldots
& \om^{10} \\ 1& \om^2&\om^4 & \ldots & \om^{20} \end{pmatrix} =
\begin{pmatrix} 1&2 &4 &
8&16&9&18&13&3&6& 12 \\ 
1 & 4 &
16 & 18&3&12&2&8&9&13 & 6 \end{pmatrix}$.

Similarly $2$-error correcting pairs may be obtained for any code
generated by \\ $\{e_i,
e_{i+1},e_{i+2},e_{i+3},e_{i+4},e_{i+5},e_{i+6}\}$. (The suffices should
be taken $\mod 11$.) For the code generated by
$\{e_0,e_2,e_4,e_6,e_{10},e_1\}$ (which have difference of $2$ in the
consecutive suffices) the pair $(U,V)$ with $U=\langle e_0, e_2,e_4 \rangle,
  V=\langle e_2,e_4 \rangle$ is a $2$-error correcting pair. 

In a similar manner $3$-error correcting pairs may be obtained for the
$(11,5,7)$ code generated by any $\{e_i,e_{i+1}, e_{i+2},
e_{i+3},e_{i+4}\}$ or more generally for any $\{e_i,e_{i+j}, e_{i+2j},
e_{i+3j},e_{i+4j}\}$ with $1\leq j \leq 10$.
 For example if $C = \langle e_0,e_2,e_4,e_6,e_8\rangle$ then
    $U=\langle e_0,e_2,e_4,e_6 \rangle, V= \langle e_2,e_4,e_6\rangle$
   constitute a $3$-error correcting pair $(U,V)$.

%\subsection{More examples ..} 


\begin{thebibliography}{99}

\bibitem{idemrank} Oskar M. Baksalary, Dennis S. Bernstein, G\"otz
  Trenkler, ``On the equality between rank and trace of an idempotent
  matrix'', Applied Mathematics and Computation, 217, 4076-4080, 2010.

\bibitem{blahut} Richard E.\ Blahut, {\em Algebraic Codes for data
  transmission}, Cambridge University Press, 2003.
\bibitem{koetter} I. Duursma, R. K\"otter, ``Error-locating pairs for
  Cyclic Codes'', IEEE Transactions in
  Information Theory, 40, 1108-1121, 1994. 
\bibitem{isaacs} R. J. Evans and I. M. Isaacs, ``Generalized
  Vandermonde determinants and roots of unity of prime order'',
  Proc.\ of Amer.\ Math.\ Soc.\, 58, 51-54, 1977.
\bibitem{simple} P. E. Frenkel, `` Simple proof of Chebotar\"ev's
 theorem on roots of unity'', arXiv:math/0312398.
\bibitem{gap} `GAP -- Groups, Algorithms and Programming',
  www.gap-system.org
\bibitem{isaacs1} Daniel Goldstein, Robert M. Guralnick, and
  I. M. Isaacs, ``Inequalities for finite group permutation modules'',
  Trans.\ Amer.\ Math.\ Soc.\ , 10, 4017-4042, 2005.

\bibitem{hur1} Paul Hurley and Ted Hurley, ``Codes from zero-divisors
  and units in group rings'', Int. J. Inform. and Coding Theory, 1
  (2009), 57-87.
\bibitem{hur2} Paul Hurley and Ted Hurley, ``Block codes from matrix
  and group rings'', Chapter 5, 159-194, in {\em Selected Topics in
    Information and Coding Theory} eds. I. Woungang, S. Misra,
  S.C. Misma, World Scientific 2010.
\bibitem{hur3} Ted Hurley, ``Group rings and rings of matrices'',
  Inter. J. Pure \& Appl. Math., 31, no.3, 2006, 319-335.
%\bibitem{hur4} Ted Hurley, ``Self-dual, dual-containing and related
%	quantum codes using
 % group rings'', arxiv:

\bibitem{hur5} Barry Hurley and Ted Hurley, ``Paraunitary matrices'',
  arXiv:1205.0703.
 
\bibitem{seh} C\'esar Milies \& Sudarshan Sehgal, {\em An introduction
  to Group Rings}, Klumer, 2002.
\bibitem{pak} F.\ Pakovich, ``A remark on the Chebotarev Theorem about
  roots of unity'', Integers: Elec.\ J.\ of Combinatorial Number
  Theory, 7, \# A18, (2 pages), 2007.
\bibitem{pell} R. Pellikaan, ``On decoding by error location and
  dependent sets of error positions'', Discrete Math., Vol.\ 106/107,
  369-381, 1992. 

\bibitem{steven} P. Stevenhagen and H. W. Lenstra, Jr., ``Chebotar\"ev
  and his density theorem'', Math.\ Intell.\, 18, 26-37, 1996.

\bibitem{tao} Terence Tao, `` An uncertainty principle for groups of
  prime power order'', Math.\ Res.\ Lett., 12, no.\ 1, 121-127, 2005.
\end{thebibliography}
\end{document}